\documentclass[final,5p,times,number]{elsarticle}
\usepackage{amsmath,amssymb,amsfonts,amsthm,mathtools}

\newcommand{\C}{\mathbb{C}}
\newcommand{\R}{\mathbb{R}}
\newcommand{\id}{\mathrm{id}}

\newtheorem{theorem}{Theorem}[section]
\newtheorem{lemma}[theorem]{Lemma}
\newtheorem{corollary}[theorem]{Corollary}
\theoremstyle{definition}

\theoremstyle{remark}

\numberwithin{equation}{section}

\begin{document}

\begin{frontmatter}

\title{A symmetry formula  for correlation functions in the \\ superintegrable chiral Potts spin chain }

\author[first]{Haoran Zhu}
\ead{zhuh0031@e.ntu.edu.sg}

\affiliation[first]{organization={School of Physical and Mathematical Sciences, Nanyang Technological University},
            addressline={21 Nanyang Link},
            city={Singapore},
            postcode={637371},
            country={Singapore}}

\begin{abstract}
We prove an exact finite-volume symmetry formula for two-point functions in the periodic $N$-state superintegrable chiral Potts spin chain.
We show that, for every chain length $L$ and every simultaneous eigenvector of the Hamiltonian and the one-site translation operator, the correlations satisfy $\langle Z_0^r Z_R^{\dagger r}\rangle^*=\langle Z_0^r Z_{L-R}^{\dagger r}\rangle$ for $1\leqslant r\leqslant N-1$. Hence, whenever $L$ is even, the midpoint correlation $\langle Z_0^r Z_{L/2}^{\dagger r}\rangle$ is real. Then we generalise the three-state chain case to arbitrary $N$ and to every translation eigensector. This resolves a conjecture of Fabricius and McCoy.
\end{abstract}

\begin{keyword}
superintegrable chiral Potts spin chain \sep Onsager algebra \sep correlation functions \sep translation symmetry
\MSC[2020] 82B20 \sep 82B23 \sep 81R12
\end{keyword}

\end{frontmatter}

\section{Introduction}

The superintegrable chiral Potts spin chain furnishes an $N$-state generalisation of the Ising chain with the same underlying Onsager algebra. Onsager's 1944 solution of the two-dimensional Ising model introduced the algebraic structure now bearing his name \citep{Onsager1944}. In 1985, von Gehlen and Rittenberg discovered a family of $\mathbb{Z}_N$-symmetric quantum chains with infinitely many conserved charges \citep{vGehlenRittenberg1985}, and Baxter identified the superintegrable chiral Potts case as a particularly tractable representative \citep{Baxter1988}. Explicit finite-size computations revealed the characteristic Ising-like square-root form of the spectrum \citep{AlbertiniMcCoyPerkTang1989}, and this form was later shown to follow directly from Onsager's algebra \citep{Davies1990,Davies1991,DateRoan2000}. Important progress on eigenvectors, Onsager sectors, and related algebraic identities was subsequently made by Au-Yang and Perk, by Nishino and Deguchi, and by Roan \citep{AuYangPerk2008,AuYangPerk2009,AuYangPerk2010,AuYangPerk2009Loop,NishinoDeguchi2008,Roan2010Eigenvectors}. Progress on spin-operator matrix elements and on the reconstruction and form factors of local operators was later obtained by Iorgov \emph{et al.} and by Grosjean, Maillet, and Niccoli \citep{IorgovEtAl2010,GrosjeanMailletNiccoli2015}. For Onsager-type algebraic perspectives, please see, for example, \citep{terwilliger-2021-cmp,terwilliger-2021-npb,terwilliger-2022-npb,lu-ruan-wang-2023}.

For a periodic chain of length $L$, the basic observables are the two-point functions
\[
\langle Z_0^r Z_R^{\dagger r}\rangle,
\qquad
1\leqslant r\leqslant N-1,
\qquad
R\in\{0,1,\dots,L-1\}.
\]
By contrast with the Ising case, where free-fermion and determinant methods give detailed control of correlations and spontaneous magnetisation \cite{Kaufman1949,KaufmanOnsager1949,MontrollPottsWard1963,Yang1952} (see also free-fermion approach in~\cite{barouch-mccoy-1971,cunden-majumdar-oconnell-2019,greaves-jing-zhu-2026}), finite-distance correlation functions in the chiral Potts setting remain far less explicit. Fabricius and McCoy computed these ground-state correlations for the three-state superintegrable chain at lengths $L=3,4,5$ \citep{FabriciusMcCoy2010}. From their finite-size data they proposed a conjectural form for the nearest-neighbour correlation, and in their concluding discussion, they pointed to a second striking phenomenon: for even chain length the half-chain correlation
\[
\langle Z_0 Z_{L/2}^{\dagger}\rangle
\]
appears to be real \citep{FabriciusMcCoy2010}. 


The main aim of the present paper is to show that the finite-volume reality statement is governed by a symmetry principle. We prove an exact finite-volume identity for two-point functions in every translation eigensector of the periodic superintegrable chiral Potts chain. 



We write $N$ for the number of spin states and $L$ for the chain length (thus our $L$ corresponds to the symbol $\mathcal{N}$ in \citep{FabriciusMcCoy2010}). Let
\[
\omega=e^{2\pi i/N},
\]
and let $Z,X\in M_N(\C)$ be the standard Weyl operators,
\begin{align*}
Z_{ab}&=\omega^a\delta_{ab},
\qquad
X_{ab}=\delta_{a,b+1},
\\
ZX&=\omega XZ,
\qquad
Z^N=X^N=\id.
\end{align*}
On the periodic chain $(\C^N)^{\otimes L}$ we write $Z_j$ and $X_j$ for the corresponding local operators at site $j$, with all indices understood modulo $L$. The periodic superintegrable chiral Potts Hamiltonian \citep{vGehlenRittenberg1985,Baxter1988,FabriciusMcCoy2010} is
\begin{equation}\label{eq:Hamiltonian}
H=A_0+\lambda A_1,
\end{equation}
where
\begin{align}
A_0&=-\sum_{j=0}^{L-1}\sum_{r=1}^{N-1}
\frac{e^{\frac{i\pi(2r-N)}{2N}}}{\sin(\pi r/N)}\, Z_j^r Z_{j+1}^{\dagger r},
\nonumber\\
A_1&=-\sum_{j=0}^{L-1}\sum_{r=1}^{N-1}
\frac{e^{\frac{i\pi(2r-N)}{2N}}}{\sin(\pi r/N)}\, X_j^r.
\label{eq:A0A1}
\end{align}
For real $\lambda$, this is the finite-volume Hamiltonian considered in \citep{FabriciusMcCoy2010}. Let $T$ denote the unitary one-site translation operator. Since $H$ is a periodic sum, it commutes with $T$; hence each eigenspace of $H$ admits a basis of simultaneous eigenvectors of $H$ and $T$.

Our main result is the following theorem.

\begin{theorem}\label{thm:main}
Let $|\psi\rangle$ be a normalised simultaneous eigenvector of the Hamiltonian $H$ in \eqref{eq:Hamiltonian}--\eqref{eq:A0A1} and the one-site translation operator $T$. For $1\leqslant r\leqslant N-1$ and $R\in \mathbb{Z}/L\mathbb{Z}$, define
\[
\rho_r(R):=\langle \psi|Z_0^r Z_R^{\dagger r}|\psi\rangle.
\]
Then
\begin{equation}\label{eq:main-symmetry-intro}
\rho_r(R)^*=\rho_r(-R)
\end{equation}
for every $R\in \mathbb{Z}/L\mathbb{Z}$. In particular, if $L$ is even, then
\begin{equation}\label{eq:halfway-real-intro}
\rho_r(L/2)\in\R.
\end{equation}
\end{theorem}



\section{Symmetry formula and reality conjecture} \label{sec:proof} 

We begin with the one-site translation operator. Let $T$ be the unitary operator on $(\C^N)^{\otimes L}$ characterised by
\begin{equation}\label{eq:translation-covariance}
T Z_j T^{-1}=Z_{j+1},
\qquad
T X_j T^{-1}=X_{j+1},
\qquad
T^L=\id,
\end{equation}
with all site labels understood modulo $L$. Since the Hamiltonian \eqref{eq:Hamiltonian}--\eqref{eq:A0A1} is a periodic sum over the sites, we have
\begin{equation}\label{eq:HTcommute}
THT^{-1}=H,
\qquad\text{equivalently}\qquad
[H,T]=0.
\end{equation}
Hence, each eigenspace of $H$ admits an orthonormal basis of simultaneous eigenvectors of $H$ and $T$.

The following elementary lemma is needed for the proof of the main theorem.

\begin{lemma}\label{lem:translation-invariance}
Let $|\psi\rangle$ be a normalised eigenvector of $T$.
Then, for every operator $\mathcal O$ on $(\C^N)^{\otimes L}$ and every integer $m$,
\begin{equation}\label{eq:translation-expectation}
\langle\psi|T^{-m}\mathcal O T^m|\psi\rangle
=
\langle\psi|\mathcal O|\psi\rangle.
\end{equation}
\end{lemma}

\begin{proof}
Since $T|\psi\rangle=\tau|\psi\rangle$, we have $T^m|\psi\rangle=\tau^m|\psi\rangle$ and, by taking adjoints,
\[
\langle\psi|T^{-m}=\overline{\tau}^{\,m}\langle\psi|.
\]
Therefore
\[
\langle\psi|T^{-m}\mathcal O T^m|\psi\rangle
=
\overline{\tau}^{\,m}\tau^m\langle\psi|\mathcal O|\psi\rangle
=
\langle\psi|\mathcal O|\psi\rangle,
\]
as required.
\end{proof}

We now prove the symmetry relation \eqref{eq:main-symmetry-intro}. The key point is that complex conjugation reverses the order of the local operators, while translation transports the resulting separation to its opposite on the ring.

\begin{proof}[Proof of Theorem~\ref{thm:main}]
Fix $r\in\{1,\dots,N-1\}$ and $R\in\mathbb Z$. By definition,
\[
\rho_r(R)=\langle \psi|Z_0^r Z_R^{\dagger r}|\psi\rangle,
\]
where the site label $R$ is read modulo $L$. Taking complex conjugates and using $(AB)^\dagger=B^\dagger A^\dagger$, we obtain
\begin{equation}\label{eq:conj-step}
\rho_r(R)^*
=
\langle \psi|Z_R^r Z_0^{\dagger r}|\psi\rangle.
\end{equation}
Apply Lemma~\ref{lem:translation-invariance} with $\mathcal O=Z_R^r Z_0^{\dagger r}$ and $m=R$. Then
\begin{equation}\label{eq:translate-step}
\rho_r(R)^*
=
\langle\psi|T^{-R}(Z_R^r Z_0^{\dagger r})T^R|\psi\rangle.
\end{equation}
Using \eqref{eq:translation-covariance}, we compute
\[
T^{-R}Z_R^rT^R=Z_0^r,
\qquad
T^{-R}Z_0^{\dagger r}T^R=Z_{-R}^{\dagger r},
\]
again with site labels understood modulo $L$. Substituting this into \eqref{eq:translate-step} gives
\[
\rho_r(R)^*
=
\langle\psi|Z_0^r Z_{-R}^{\dagger r}|\psi\rangle
=
\rho_r(-R).
\]
This proves \eqref{eq:main-symmetry-intro}.

If $L$ is even and $R=L/2$, then $-R\equiv R \pmod L$. Hence
\[
\rho_r(L/2)^*=\rho_r(L/2),
\]
so $\rho_r(L/2)$ is real. This establishes \eqref{eq:halfway-real-intro}.
\end{proof}

It is convenient to record the original three-state statement in the form conjectured by Fabricius and McCoy.

\begin{corollary}[Fabricius--McCoy conjecture]\label{cor:FM}
Assume $N=3$ and let $L$ be even. Let $|\psi_0\rangle$ be a normalised finite-volume ground state chosen to be an eigenvector of the one-site translation operator. Then
\[
\langle\psi_0|Z_0 Z_{L/2}^{\dagger}|\psi_0\rangle\in\R.
\]
\end{corollary}

\begin{proof}
Since $[H,T]=0$, the finite-volume ground-state eigenspace contains an orthonormal basis of simultaneous eigenvectors of $H$ and $T$. The claim is the case $N=3$ and $r=1$ of Theorem~\ref{thm:main}.
\end{proof}

\bigskip
\noindent{\bf Acknowledgements
} We would like to thank the support of the NTU Research Scholarship (RSS). The author also benefits from the discussion with K. Wang.

\end{document}